\documentclass[%
reprint,
superscriptaddress,
frontmatterverbose, 
preprintnumbers,
longbibliography,
amsmath,
amssymb,
aps,
pra,
notitlepage,
]{revtex4-1}

\usepackage{graphicx}
\usepackage{dcolumn}
\usepackage{bm}
\usepackage{hyperref}
\usepackage{subfigure}
\usepackage{float}

\usepackage{comment}

\usepackage{mathrsfs}

\usepackage{xcolor}
\definecolor{darkgreen}{rgb}{0.1,0.5,0.1}

\hypersetup{
colorlinks=true,
linkcolor=blue,
filecolor=blue,
citecolor=blue,  
urlcolor=blue,
}

\usepackage{amsthm}
\theoremstyle{plain}
\newtheorem{thm}{Theorem}

\newtheorem{cor}[thm]{Corollary}

\theoremstyle{definition}
\newtheorem{defn}[thm]{Definition}

\newtheorem*{rem}{Remark}

\newcommand{\ket}[1]{|#1\rangle}
\newcommand{\bra}[1]{\langle #1|}
\newcommand{\bracket}[2]{\langle #1|#2\rangle}
\newcommand{\ketbra}[2]{|#1\rangle\!\langle #2|}
\newcommand{\dm}[1]{\ketbra{#1}{#1}}

\newcommand{\stab}{\mathrm{STAB}}
\newcommand{\dmax}{D_{\mathrm{max}}}
\newcommand{\dmin}{D_{\mathrm{min}}}

\newcommand{\ddmax}{\mathfrak{D}_{\mathrm{max}}}
\newcommand{\ddmin}{\mathfrak{D}_{\mathrm{min}}}

\newcommand{\ba}{\begin{eqnarray}}
\newcommand{\ea}{\end{eqnarray}}

\DeclareMathOperator{\Tr}{Tr}
\newcommand{\tr}{\operatorname{Tr}}

\newcommand{\1}{\openone}

\newcommand{\red}[1]{{\color{black}#1}}
\newcommand{\zwnew}[1]{ { \color{black}  #1 }}
\newcommand{\zwedit}[1]{ { \color{black}  #1 }}
\newcommand{\newref}[1]{{\color{black}#1}}

\begin{document}


\title{Many-body quantum magic}

\author{Zi-Wen Liu}
\email{zliu1@perimeterinstitute.ca}
\affiliation{Perimeter Institute for Theoretical Physics, Waterloo, Ontario  N2L 2Y5, Canada}

\author{Andreas Winter}
\email{andreas.winter@uab.cat}
\affiliation{ICREA---Instituci\'o Catalana de Recerca i Estudis Avan\c{c}ats, Pg.~Lluis Companys, 23, ES-08001 Barcelona, Spain} 
\affiliation{F\'{\i}sica Te\`{o}rica: Informaci\'{o} i Fen\`{o}mens Qu\`{a}ntics, 
Departament de F\'{\i}sica, Universitat Aut\`{o}noma de Barcelona, ES-08193 Bellaterra (Barcelona), Spain}

\date{\today}

\begin{abstract}
Magic (non-stabilizerness) is a necessary but ``expensive'' kind of  ``fuel'' to drive universal fault-tolerant quantum computation. 
To properly study and characterize the origin of quantum ``complexity'' in computation as well as physics,
it is crucial to develop a rigorous understanding of  the quantification of magic.
Previous studies of magic mostly focused on small systems and largely relied on the discrete Wigner formalism (which is only well behaved in odd prime power dimensions). 
Here we present an initiatory study of the magic of genuinely many-body quantum states that may be strongly entangled, with focus on the important case of many qubits, at a quantitative level.
We first address the basic question of how ``magical'' a many-body state can be,
and show that the maximum magic of an $n$-qubit state is essentially $n$,
simultaneously for a range of ``good'' magic measures. As a corollary, 
 the resource theory of magic has asymptotic golden currency states.    We then show that, in fact, almost all $n$-qubit pure states have magic of nearly $n$.    
In the quest for explicit, scalable cases of highly entangled states whose magic can be understood, we connect the magic of hypergraph states with the second-order nonlinearity of their underlying Boolean functions. 
Next, we go on and investigate many-body magic in practical and physical contexts.
We first consider a variant of measurement-based quantum computation (MBQC) where the client is restricted to Pauli measurements, in which magic is a necessary feature of the initial ``resource'' state. We show that $n$-qubit states with nearly $n$ magic, or indeed almost all states, cannot supply nontrivial speedups over classical computers.
We then present an example of analyzing the magic of ``natural'' condensed matter systems of physical interest.  We apply the Boolean function techniques to derive explicit bounds on the magic of certain representative two-dimensional symmetry-protected topological (SPT) states, and comment on possible further connections between magic and the quantum complexity of phases of matter.  
\end{abstract}


\maketitle


\section{Introduction}
\label{sec:intro}
The Clifford group and the closely associated stabilizer formalism \cite{Gottesman97:thesis,Gottesman98} are a central notion in quantum computation and many related areas. 
The Clifford group on $n$ qubits is defined as the normalizer of the $n$-qubit Pauli (aka discrete Weyl) 
group, and can be generated by the Hadamard gate, the $\pi/4$ phase gate, and the controlled-NOT gate.   Then the quantum states that can be prepared by acting Clifford operations on a canonical trivial state are called stabilizer states. 
The celebrated Gottesman--Knill theorem \cite{Gottesman98,NielsenChuang} indicates that a quantum computation with only Clifford or stabilizer components, despite being capable of generating as much entanglement as possible and exhibiting very rich structures, can be efficiently simulated by a classical computer.  In other words, non-stabilizer features are needed in order to enable universal quantum computation and achieve the desired quantum computational advantages.  
Moreover, the Clifford operations are generally considered relatively cheap in fault-tolerant quantum computation 
by virtue of the widely studied stabilizer quantum error correction codes \cite{Gottesman97:thesis,Veitch14:magic_rt,6006592}, which are generated by Clifford operations themselves. 
All in all, the non-stabilizerness, commonly referred to as \emph{magic}, 
represents a particular key resource for quantum computation, both 
from a fundamental and a practical point of view.
A rigorous, quantitative understanding of magic would  play key roles in the  study of quantum computation and complexity in many ways. 
For example, a direction of great recent interest is to link magic measures to the costs of classical simulation algorithms \cite{BravyiSmithSmolin16:stabrank,Bravyi2019simulationofquantum,HowardCampbell17,WangWilde:magic-channel,SeddonCampbell19:magic_op,SeddonCampbell19:magic_op,Seddon20:speedups}. 
Indeed, the quantification of other important quantum resource features such as entanglement \cite{PhysRevLett.78.2275,RevModPhys.81.865} and coherence \cite{PhysRevLett.113.140401,RevModPhys.89.041003} has been a characteristic research line of quantum information, which help understand and characterize ``quantumness'' in various scenarios.

Previous studies on magic measures (e.g.,~Refs.~\cite{Veitch14:magic_rt,HowardCampbell17,BravyiSmithSmolin16:stabrank,Bravyi2019simulationofquantum,Heinrich2019robustnessofmagic,LiuBuTakagi19:oneshot,WangWildeSu,WangWilde:magic-channel,SeddonCampbell19:magic_op,Seddon20:speedups,Heimendahl2021stabilizerextentis}) largely focused on small or weakly correlated systems. Little is known about magic in entangled many-body systems.   In particular,  the number of stabilizer states grows very rapidly and  their geometric structures become highly complicated as one increases the size of the system, which makes the calculation or even numerical analysis of magic measures on large states difficult in general. 
Some fundamental mathematical questions that we would like to understand include the following. \red{How ``magical'' can many-body quantum states be?  How much magic do generic states typically contain?} How can we calculate the magic of many-body states?
Moreover, much of our existing understanding on magic relies on the discrete Wigner formalism \cite{WOOTTERS19871,Gross:Hudson} (see, e.g.,~Refs.~\cite{Veitch12:negativity,Veitch14:magic_rt,WangWildeSu,WangWilde:magic-channel}), 
which is easier to deal with and allows for simple magic measures, but usually only well defined and connected to the magic theory for qudits of odd prime dimensions. Given the clear importance of qubit systems, we would like to have a systematic general theory of magic.  

Another important motivation comes from a physics perspective. A theme of many-body physics is to characterize or classify different phases of matter according to their  physical features 
such as symmetry, magnetism, or superconductivity, through certain  order parameters.  A new perspective that has drawn great interest is to investigate the ``quantum complexity'' of phases, which is encoded in, e.g.,~the cost of probing or simulating them and their computational power.
Important relevant topics include, among others, the computational universality in measurement-based quantum computing (MBQC) \cite{RaussendorfBriegel01,RaussendorfBrowneBriegel03,Wei18:mbqc,Raussendorf19:universal}, sign problems for Monte Carlo methods (see, e.g.,~Refs.~\cite{gubernatis_kawashima_werner_2016,DBLP:journals/qic/BravyiDOT08,Hastings16:signproblem}), etc.  
Can we find some ``order parameters'' to probe these computation-related features of many-body systems?  Given the fundamental connection between magic and quantum computation, exploring the roles of many-body magic would be a promising direction to go.

In this work, we investigate magic in general many-body quantum systems at a quantitative level from both mathematical and physical perspectives.     We first give a systematic introduction of magic measures induced from general resource theories, and discuss their relations with several other known  magic measures.  Importantly, we give an argument about a range of ``good'' resource measures, showing that any measure that satisfies certain consistency conditions  in terms  of state transformation is sandwiched in between the min-relative entropy of resource and the  free robustness. That is, the min-relative entropy of magic and the  free robustness of magic are in some sense the ``extremes'' of the family of (suitably regularized) magic measures.
We show that the roof values of such consistent magic measures on an $n$-qubit state are essentially $n$, \red{implying that the resource theory of magic has asymptotic currency states that play important roles in the study of resource manipulation.}
We also show that the magic measures typically take value very close to $n$ on an $n$-qubit pure state, which resembles the situation of entanglement (the well-known Page's theorem and its variants \cite{Page93,LLOYD1988186,HaydenLeungWinter,LLZZ18:design,Liu2018:jhep}).    Then, we turn to the quest for explicit methods for analyzing the many-body magic of certain states.  In this work we consider the family of hypergraph states \cite{Rossi13:hypergraph}, \red{which are widely relevant in MBQC \cite{MillerMiyake16,TakeuchiMorimaeHayashi19,PhysRevA.99.052304}, quantum error correction \cite{Lyons_2017,2017arXiv170803756B}, quantum many-body physics \cite{LevinGu12,MillerMiyake16}, etc.  As will be discussed in more detail, hypergraph states constitute a natural and rich class of magic states generated by $C^{k-1}Z$ (multi-controlled-$Z$) gates, which are diagonal gates in the $k$-th level of the Clifford hierarchy (the first two levels of which constitute the Clifford group) \cite{GottesmanChuang,CuiGottesmanKrishna} closely associated with degree-$k$ Boolean functions.}
We find that the magic of hypergraph states can be understood by analyzing the second-order nonlinearity of their underlying Boolean functions {or Reed--Muller codes (which is a problem of great interest also in coding theory and cryptography)}, thus establishing connections between many-body magic, Boolean function analysis, and coding theory.  
Next, we make some  observations about many-body magic in regard to quantum computation and condensed matter physics. First, the quantum computational power of many-body states  manifests itself in the  MBQC \cite{RaussendorfBriegel01,RaussendorfBrowneBriegel03} setting,  one of the standard models of quantum computation.  Here we suggest considering a practical variant of MBQC that we call Pauli MBQC, where the client is allowed to make Pauli measurements only and thus the magic of the quantum computation is completely supplied by the resource state prepared offline.  We show that  many-body states with nearly $n$ magic, and indeed almost all states, cannot support universal quantum computation, or even nontrivial speedups over classical computers.   That is, akin to the situation that most states are too entangled to be useful for conventional MBQC \cite{GrossFlammiaEisert08,BremnerMoraWinter08}, here, most  states are ``too magical'' to be useful in the Pauli MBQC model, which highlights the curious phenomenon that \red{the computational power is not simply determined by the amount of computational resource}.
This is a no-go result in the high-magic regime, and eventually we would like to find more fine-grained relations between magic and computational power.  
Then, we take a first look at the magic of many-body systems of interest in condensed matter physics, \newref{extending the efforts of applying resource theory to physics in the contexts of, e.g.,~thermodynamics (see, e.g.,~Ref.~\cite{YungerHalpern2017}) and asymmetry (see, e.g.,~Refs.~\cite{Marvian14,ZLJ20})}.    A particularly interesting case that we focus on is the symmetry-protected topological (SPT) phases, where magic is expected to be a key physical feature,  especially in beyond one dimension \cite{EKLH20,MillerMiyake16}. 
As a demonstration, we employ the Boolean function techniques to derive explicit bounds on the magic of certain representative SPT ground states defined on different two-dimensional lattices, based on their hypergraph state form. A general conclusion is that the magic of such SPT states is {rather weak compared to typical entangled states} (although generically necessary and robust \cite{EKLH20}), which goes hand in hand with their short-range entangled feature, and is consistent with the Pauli MBQC universality known for certain cases.  
Lastly, we discuss possible further relations of many-body magic and the many facets of the quantum complexity of phases of matter.  
We hope to raise further interest in the characterization and application of many-body magic in condensed matter physics, and stimulate further explorations into the connections between quantum computation, complexity, and physics.  
(The quantitative behavior of magic in different many-body systems of physical interest is recently independently investigated in Refs.~\cite{WhiteCaoSwingle,Sarkar_2020}.)


\section{{Magic:\protect\\ resource theory and measures}}
\label{sec:magic}

Here we review the magic measures we mainly consider in this paper, which are rooted in the resource theory framework, 
and summarize their relations with other useful measures studied in 
the literature.

We first formally define the notation.  The Clifford group on $n$-qubits  $\mathcal{C}_n$  is defined as the normalizer of the  $n$-qubit Pauli group $\mathcal{P}_n$ composed  of tensor products of $I,X,Y,Z$ on $n$  qubits with phases $\pm 1, \pm i$:
\[
\mathcal{C}_n = \{U:  U W U^\dagger \in \mathcal{P}_n, \forall W\in\mathcal{P}_n  \}.
\]
Then the pure stabilizer states are generated by Clifford group elements acting on the trivial computational basis state $\ket{0}^{\otimes n}$.
We denote by $\stab$ the convex hull of all stabilizer states, whose 
extreme points are 
precisely the pure stabilizer states. 
\red{Also denote by $\mathcal{S}$ the set of all states.
Then, $\stab$ is a convex polytope inside $\mathcal{S}$, in dimension $4^n-1$ and with 
$2^{\Theta(n^2)}$ vertices, for $n$ qubits.} It is a highly symmetric 
object, with the Clifford group acting multiply transitively on the vertex set. 

\red{The Gottesman--Knill theorem \cite{Gottesman98,NielsenChuang} provides motivations for understanding the quantum computation advantages through the resource theory of magic \cite{Veitch14:magic_rt}, where the stabilizer polytope $\stab$ is considered to be the set of free states, as speedups over classical computers require states outside $\stab$.  
 } 
 \red{The set of free operations relevant in this work is the set of all $\stab$-preserving or, equivalently, magic non-generating operations, which is the maximal set of free operations that strictly contain several other choices studied before, including stabilizer protocols \cite{Veitch14:magic_rt} and completely $\stab$-preserving operations \cite{SeddonCampbell19:magic_op,Seddon20:speedups,Heimendahl20:axiomatic}.}
 
 \red{A central task in resource theories is the quantification of resource through resource measures.}
From the meta-theory of general resource theories, we have straightaway the following important measures of magic that satisfy fundamental properties such as monotonicity under  free  ($\stab$-preserving, which include Clifford) operations, faithfulness, etc.:
\begin{itemize}
   \item Min-relative entropy of magic: 
          \[
            \ddmin(\rho) = \min_{\sigma\in\stab}\dmin(\rho\|\sigma),
          \]
          with \red{the min-relative entropy}
          $D_{\min }(\rho \| \sigma) := -\log \tr\Pi_{\rho} \sigma$, \red{where $\Pi_\rho$ is the projector onto the support of $\rho$}.
          For a pure state $\ket{\psi}$,  
          \(
            \ddmin(\psi) = -\log\max_{\phi\in\stab}|\bracket{\psi}{\phi}|^2.
          \)

    \item Max-relative entropy of magic:
          \[
            \ddmax(\rho) = \min_{\sigma\in\stab}\dmax(\rho\|\sigma), 
          \]
          with \red{the max-relative entropy} $\dmax(\rho\|\sigma) := \log \min \{\lambda: \rho \leq \lambda \sigma\}$, 
          \red{where the matrix inequality $\rho \leq \lambda \sigma$ means that $\lambda \sigma - \rho$ is positive semidefinite.}
          This measure is also known as log-generalized robustness,
          $\ddmax(\rho) = \log(1+R_g(\rho))$, where
          \[
            \qquad R_g(\rho) = \min s\geq 0 \text{ s.t. } \rho \in (1\!+\!s)\,\stab - s\,\mathcal{S}.
          \]
\red{Here the subscript ``$g$'' is a label for ``generalized robustness,'' signifying its difference with the free robustness measure that will also be discussed.}
 
    \item Free robustness of magic: 
          \[
            \qquad R(\rho) = \min s\geq 0 \text{ s.t. } \rho \in (1\!+\!s)\,\stab - s\,\stab.
          \]
          The log-free robustness is $LR(\rho) =\log (1+R(\rho))$.
\end{itemize}

\red{
We now provide some general intuitions for these measures. Note that $\dmin,\dmax$ are both divergences that characterize the ``distance'' between two states in a certain sense, so $\ddmin,\ddmax$ measure the minimum such distances to the set of free states ($\stab$).  In particular, roughly speaking, $\dmin,\dmax$ respectively correspond to the minimum and maximum ``extremes'' of the widely studied family of quantum R\'enyi relative entropies (see, e.g.,~Refs.~\cite{Datta_2014,qrenyi} for more comprehensive discussions on quantum R\'enyi relative entropies). Moreover, the robustness-type measures $\ddmax$ and $LR$ are intuitively related to the amount of ``noise'' needed to erase the resource (magic) upon  mixture, where $\ddmax$ considers all possible noise states and $LR$ considers free noise states (from $\stab$).  Operationally, these measures and their variants  play fundamental roles in the characterization of the rates of resource conversion, in the practical one-shot (finite resource) setting
\cite{LiuBuTakagi19:oneshot,Zhou_2020}.  
It is worth noting that the theory of magic is a particularly interesting one where these three types of measures are all non-trivially defined and inequivalent at the same time, in contrast to many other important resource theories (for example, in coherence theory, $LR$ is not even finite \cite{Napoli:RoC}, and in bipartite entanglement theory, $\ddmax$ and $LR$ coincide  \cite{Steiner2003robustness,Harrow2003robustness}).  

Note that for any state $\rho$, it holds that 
\begin{equation}
    \ddmin(\rho) \leq \ddmax(\rho) \leq LR(\rho).
\end{equation}
In Appendix~\ref{sec:range}, we give an argument supporting that  ``good'' resource (magic) measures satisfying certain consistency conditions induced from state transformability are sandwiched between  the min-relative entropy of resource ($\ddmin$) and the free robustness ($LR$). In the following, when we say a magic measure $f$ is ``consistent'' we mean that $\ddmin \leq f \leq LR$.  


These measures have close relations with several other recently studied important magic measures arising from various contexts: 
}

\red{
\begin{itemize}
    \item Stabilizer extent \cite{Bravyi2019simulationofquantum}: 
    The stabilizer extent for pure state $\psi$ is defined as $\xi(\psi) := \min \left(\sum |c_\phi|\right)^2 \text{ s.t. } \ket{\psi} = \sum_{\phi\in\stab} c_\phi \ket{\phi} $.
          It holds that $\xi(\psi) = 2^{\ddmax(\psi)} = 1+R_g(\psi)$.   Note that for general states including mixed ones, we have the convex roof extension of stabilizer extent and the dyadic negativity \cite{Seddon20:speedups}, both of which reduce to $\xi$ on pure states.  
    \item Stabilizer fidelity \cite{Bravyi2019simulationofquantum}: 
         The stabilizer fidelity  for pure state $\psi$ is defined as $F(\psi) := \max_{\phi\in\stab}|\bracket{\psi}{\phi}|^2$. It holds that $F(\psi) = 2^{-\ddmin(\psi)}$.
    \item Stabilizer rank  \cite{BravyiSmithSmolin16:stabrank,BravyiGosset,Bravyi2019simulationofquantum}: 
    For pure state $\psi$, the exact stabilizer rank  is defined as $\chi(\psi):=\min k\in\mathbb{N} \text{ s.t. } \ket{\psi} = \sum_{i=1}^{k} c_i \ket{\phi_i} , \phi_i\in\stab$, and the approximate or ``smooth'' version is given by $\chi^\epsilon(\psi) := \min \chi(\psi') \text{ s.t. }\|\psi-\psi'\|\leq\delta$.
         We have the bound 
          $\chi^\epsilon(\psi)\leq 1 + \xi(\psi)/\epsilon^2 = 1 + 2^{\ddmax(\rho)}/\epsilon^2$.
    \item Wigner negativity and mana 
    \cite{Veitch14:magic_rt,PashayanWallmanBartlett15:negativity}: For state $\rho$ in odd prime power dimensions, we can define magic measures such as mana $\mathscr{M}(\rho)$ based on the negative values of the discrete Wigner representation. We have the bound
            $\mathscr{M}(\rho)\leq
                    LR(\rho)+1$.  See Appendix~\ref{sec:negativity} for detailed definitions and proofs.
\end{itemize}
}

In the present paper, we focus on multi-qubit systems, but it is worth noting that the Pauli group, and thus the Clifford group as its normalizer, generalize to arbitrary local dimension $d$, the theory being algebraically most satisfying if $d$ is a prime power. {The appendix includes some considerations for odd prime power dimensions.} 
In odd dimensions, a necessary (but for mixed states not sufficient) condition for a state being in $\stab$ is that it has a non-negative discrete Wigner function \cite{Gross:Hudson,Veitch12:negativity}. The so-called mana measures how much negativity the Wigner function has \cite{Veitch14:magic_rt}. Building on this, the so-called thauma measures \cite{WangWildeSu} are also defined by minimum divergences (such as the max- and min-relative entropies above), but for odd prime power dimensions, relative to positive semidefinite matrices with non-negative discrete Wigner function.

\section{Behaviors of magic measures}
\label{sec:bounds}

As noted, the behaviors of magic measures on
general, entangled many-body states is little understood. The most fundamental questions are about their roof  and typical values, which we now study.  It is easy to see that the 
maximum value of the min- and max-relative entropies of magic over product states
(and indeed over all fully separable states) is 
\[
  \ddmax(\text{SEP}) = \ddmin(\text{SEP}) 
           = [\log(3-\sqrt{3})]n \approx 0.34n,
\]
attained on the tensor product of the ``golden state'' 
$G = \frac{1}{2}(I+\frac{X+Y+Z}{\sqrt{3}})$ \cite{,,LiuBuTakagi19:oneshot}, 
due to weak additivity.  
Note that these measures carry fundamental operational interpretations in 
terms of value in transformations.  How large can they get when we consider 
general states?

First, observe that the value of $\ddmax$ or log-generalized robustness (and so of all entropic measures) is 
capped at $n$:

\begin{thm}
  On an $n$-qubit system, 
  $\max_\rho \ddmax(\rho) \leq n$.
\end{thm}
\begin{proof}
The generalized robustness of magic is upper bounded by the generalized 
robustness of coherence, since $\stab$ contains all diagonal density 
matrices. The maximum value of the log-generalized robustness of coherence is $n$
\cite{Napoli:RoC}.
\end{proof}

The free robustness could in general be much larger than the generalized robustness or even infinite (e.g.,~in coherence theory).  But here we find that $LR$ is virtually also bounded above by $n$: 
\begin{thm}
  \label{thm:LR-max}
  For any $n$-qubit state $\rho$, $R(\rho) \leq \sqrt{2^n(2^n+1)}$. 
  Therefore, $\max_\rho LR(\rho) \leq n + 2^{-n-1}$.
\end{thm}
\begin{proof}
The free robustness is a linear program (LP), 
\begin{equation}
  1+R(\rho) = \min \sum |c_\phi| \text{ s.t. } \rho = \sum_{\phi\in\stab} c_\phi \phi,
\end{equation}
where $c_\phi$ are real coefficients.
Its dual LP is well known:
\begin{equation}
  1+R(\rho) = \max \tr\rho A \text{ s.t. } 
  \red{|\tr\phi A| \leq 1, \forall \phi\in\stab,}
    \label{eq:duallp}
\end{equation}
where the maximum runs over Hermitian matrices $A$. 
Thus, 
\begin{equation}
  \max_\rho \red{1+R(\rho)} = \max \|A\| \text{ s.t. } 
  \red{|\tr\phi A| \leq 1, \forall \phi\in\stab,}
\end{equation}
\red{where $\|\cdot\|$ denotes the operator (spectral) norm.}
We expand $A$ in the Pauli basis, $A = \sum_P \alpha_P P$, so that
\begin{equation}
  \label{eq:HS-bound}
  \|A\|^2 \leq \tr A^2 = 2^n \sum_P \alpha_P^2.
\end{equation}
On the other hand, a pure stabilizer state $\phi$ is given by an Abelian 
subgroup $G$ of the Pauli group \red{that does not contain $-\1$ or indeed any other scalars except $\1$, of maximum} cardinality $2^n$, and a character
$\chi:G\rightarrow {\pm 1}$:
\begin{equation}
  \phi = 2^{-n}\left( \1 + \sum_{P\in G\setminus\1} \chi(P)P \right).
\end{equation}
Thus, for a dual feasible $A$,
\begin{equation}
  \label{eq:A-dual-feasible}
  \tr\phi A = \alpha_{\1} + \sum_{P\in G\setminus\1} \chi(P) \alpha_P \leq 1.
\end{equation}
Now note that $\left[\sqrt{\frac{1}{2^n}}\chi(P)\right]_{P,\chi}$ is a 
unitary matrix, and so
\begin{equation}\begin{split}
  \sum_{P\in G} \alpha_P^2 
                           = 2^{-n} \sum_\chi \left( \sum_{P\in G} \chi(P)\alpha_P \right)^2 
                            \leq 1, \label{eq:sum_G}
\end{split}\end{equation}
 because of Eq.~(\ref{eq:A-dual-feasible}). 

Now, we use the fact \cite{Boykin} that the Pauli group modulo phases
is a union of $2^n+1$ stabilizer subgroups that intersect only in the
identity: $\widetilde{\mathcal{P}}_n\setminus\1 = \bigcup_{j=0}^{2^n} G_j\setminus\1$.
This allows us to obtain from the last equation, by summing over $j$,
\begin{equation}\begin{split}
  \sum_P \alpha_P^2 &\leq (2^n+1)\alpha_{\1}^2 + \sum_{P\neq\1} \alpha_P^2 \\
                    &=    \sum_{j=0}^{2^n} \sum_{P\in G_j} \alpha_P^2 
                     \leq 2^n+1. \label{eq:sum_P}
\end{split}\end{equation}
Together with Eq.\ (\ref{eq:HS-bound}), we obtain $\|A\|^2 \leq 2^n(2^n+1)$,
concluding the proof.
\end{proof}

\red{
\begin{rem} 
Observe that in the proof we did not actually use the set of all stabilizer states, only the $2^n(2^n+1)$ states from a complete set of mutually unbiased bases.  
An anonymous referee of an earlier version of this paper has pointed out that the above result holds in fact more generally for the free robustness with respect to any complex projective $2$-design (recall that a complete set of MUBs is an instance of that), and that a simpler proof can be given.  

Indeed, consider any Hermitian $A$ satisfying the constraints of the dual program (\ref{eq:duallp}) for all $\phi\in\mathcal{D}$, where $\mathcal{D}$ is the $2$-design, coming with a probability distribution $p(\phi)$. Then, 
\begin{align}
  1 &\geq \sum_{\phi\in\mathcal{D}} p(\phi) \tr\left(\dm{\phi}^{\otimes 2} A^{\otimes 2}\right) \\
    &= \frac{2}{2^n(2^n+1)} \tr \Pi_{\text{sym}}A^{\otimes 2} \\
    &= \frac{1}{2^n(2^n+1)} \tr (\1+S)A^{\otimes 2} \\
    &= \frac{1}{2^n(2^n+1)} {[(\tr A)^2 + \tr A^2]}, 
\end{align}
where $\Pi_{\text{sym}}$ is the projector onto the symmetric subspace and $S$ is the swap operator; the second line follows from the definition of $2$-design, and the last line follows from the ``swap trick'', $\tr SA^{\otimes 2} = \tr A^2$. 
So we have $\|A\|^2 \leq \|A\|_2^2 = \tr A^2 \leq 2^n(2^n+1)$, and the rest of the proof is as above.
\end{rem}  
}

Geometrically, this result indicates in a rough sense that $\stab$  occupies the whole state space quite well, so that optimizing over all states in the definition of robustness does not help much as compared to optimizing over $\stab$ only.
{While in the resource theory of entanglement, there are several studies into the relative volume of the separable states starting with Ref.~\cite{PhysRevA.58.883}, we are not aware of similar results for $\stab$.}

Now that the log-generalized robustness and log-free robustness of an  $n$-qubit state are shown to be upper bounded essentially  by $n$, as are all other measures of present interest, we turn to the question of whether there are highly magical states that approach the upper bounds.  
Here, we show that the min-relative entropy of magic (and thus all entropic measures) of a Haar-random state typically gets close to $n$, which means that, in fact, almost all states achieve nearly maximum values of all consistent magic measures. 

\begin{thm}
  Let $\ket{\psi}$ be a random $n$-qubit state drawn from the Haar measure. Then for any $n\geq 6$,
  \red{
  \begin{align}
    \Pr\left\{\ddmin(\psi) \leq \gamma  \right\} < \exp(0.54 n^2 - 2^{n-\gamma}).  \label{eq:dminbound}
\end{align}
 This implies that
  \begin{equation}
      \Pr\left\{\ddmin(\psi) < n-2\log n  - 0.63 \right\} < \exp(-n^2).
  \end{equation} } \label{thm:haar}
\end{thm}
\begin{proof}
This result is a nonasymptotic variant of Claim 2 of
Ref.~\cite{Bravyi2019simulationofquantum}.
Let $\ket{\phi}$ be any $n$-qubit state. For Haar-random $\ket{\psi}$, the probability density function of 
$\alpha=|\langle \phi | \psi\rangle|^{2}$ is given by 
$p(\alpha) = \left(2^{n}-1\right)(1-\alpha)^{2^{n}-2}$ (see, e.g.,~Refs.~\cite{Kus_1988,Zyczkowski_2000,mehta2004random}). So the cumulative distribution function is given by 
$\Pr\left\{|\langle \phi | \psi\rangle|^{2} \geq \beta\right\}=(1-\beta)^{2^{n}-1} \leq \exp(-(2^n-1)\beta)$.  

By the union bound, we have
\begin{equation}
     \Pr\left\{ \max _{{\phi} \in \operatorname{STAB}} |\langle\phi | \psi\rangle|^{2} 
                                \geq \epsilon \right\} \leq |\stab_n|\cdot  \exp(-(2^n-1)\epsilon),  \label{eq:}
\end{equation}
where $|\stab_n|$ is the cardinality of the set of $n$-qubit pure stabilizer states. It is known \cite{AaronsonGottesman04} that
\begin{equation}
    |\stab_n|=2^{n} \prod_{k=0}^{n-1}\left(2^{n-k}+1\right).
\end{equation}
It can be verified that $|\stab_n| = 2^{c_n n^2}$ with $c_n$ monotonically decreasing with $n$ (asymptotically, $|\stab_n| = 2^{(1/2 + o(1))n^2}$).  Note that $c_6 \approx 0.784$, so for $n\geq 6$,  $|\stab_n| < 2^{0.78 n^2}$. Continuing (\ref{eq:}), for $n\geq 6$ and $\epsilon>0$,
\begin{align}
     \Pr\left\{ \max _{{\phi} \in \operatorname{STAB}} |\langle\phi | \psi\rangle|^{2} 
                                \geq \epsilon \right\}  &< 2^{0.78 n^2}\cdot\exp(-(2^n-1)\epsilon)  \nonumber\\
                             &< \exp(0.54 n^2 - (2^n-1)\epsilon) \nonumber\\
                             &\leq \exp(0.54 n^2 - 2^{n+\log\epsilon}).  \nonumber
\end{align}
By the definition of $\ddmin$, the above translates to the general bound
\begin{align*}
    \Pr\left\{\ddmin(\psi) \leq \gamma  \right\} < \exp(0.54 n^2 - 2^{n-\gamma}).  \label{eq:dminbound}
\end{align*}
In order for the r.h.s.\ to be $\leq \exp(-n^2)$, we need $ 0.54 n^2 - 2^{n-\gamma} \leq -n^2$, which implies that
\begin{align*}
    \gamma \geq n-2\log n -\log(0.54+1) > n-2\log n - 0.63.
\end{align*}
\end{proof}
Note that we state the result for $n\geq 6$, simply because,  for $n<6$, it turns out that $n-2\log n - c'_n <0$, where $c'_n$ is the best corresponding constant emerging from the same derivation, so that the induced bounds are trivial.   

The situation is reminiscent to the well-studied case of entanglement, where the Haar-random values of corresponding measures are nearly maximal \cite{Page93,GrossFlammiaEisert08,LLZZ18:design,Liu2018:jhep}.  \red{Furthermore, this result readily implies that $\Omega(n)$ constant-size magic states or gates are needed to synthesize a $n$-qubit Haar-random state with overwhelming ($>1-e^{-O(n^2)}$) probability, since each of them can only supply constant magic.}

Then an interesting question is when do (approximate) unitary $t$-designs generate such nearly maximal magic with high probability.  
It is recently shown by 
Haferkamp \emph{et al.}~\cite{Haferkamp:homeopathy} that $\widetilde{O}(t^4\log(1/\epsilon))$ 
single-qubit non-Clifford gates (independent of $n$) are sufficient to form $\epsilon$-approximate unitary $t$-designs for sufficiently large $n$. 
Again, because each single-qubit gate can only generate constant magic, this result indicates that approximate designs of order at least $t=\widetilde{\Omega}(n^{1/4})$ (treating $\epsilon$ as a constant) are needed to guarantee nearly maximal, or indeed even linear magic  (in terms of all consistent magic measures).
In light of a conceptually similar result for entanglement that unitary designs \red{of} order $\approx n$ generate nearly maximal min-entanglement entropy \cite{LLZZ18:design,Liu2018:jhep}, we further conjecture that unitary $O(n)$-designs are sufficient to achieve nearly maximal $\ddmin$.
\red{A line of research of great importance and recent interest in physics is to understand the evolution of ``complexity'' in chaotic or ``scrambling'' physical dynamics through solvable models such as random quantum circuits composed of random local gates (see, e.g.,~Refs.~\cite{hayden2007,HQRY16,2019arXiv190602219H,Nahum2,Nahum1}).   
Combining with the well-known result that $t$-designs can be approximated by $O(\mathrm{poly}(t))$ random gates \cite{BrandaoHarrowHorodecki}, we expect that all divergence-based measures of entanglement and magic as probes of complexity become nearly maximal with $O(\mathrm{poly}(n))$ gates, or  $O(\mathrm{poly}(n))$ depth and time.
As an interesting comparison, note that the circuit complexity, roughly defined as the minimum number of gates needed to simulate the dynamics, is expected to grow (linearly) for exponential time (indeed, the Haar measure has exponential circuit complexity, and relatedly, a recent result formally links exponential designs to epsilon-nets of the unitary group \cite{OsmaniecSawickiHorodecki}); see, e.g.,~Refs.~\cite{Susskind14,PhysRevD.90.126007,BrownSusskind18,Susskind18,Brandao19:complexity-growth} for more detailed discussions on such phenomena and their physical relevance. To summarize, the saturation of entropic measures is expected to happen upon convergence to $O(n)$-designs that is in the polynomial time regime, whereas the more refined circuit complexity should grow for much longer (exponential) time.  
}


A closely related result is the following.
\begin{thm}\label{thm:maxdmin}
For any $n$,
\begin{equation}
    \max_{\rho\in\mathcal{S}(\mathcal{H}_2^{\otimes n})}\ddmin(\rho)> n - 2\log n + 0.96.
\end{equation}
\red{The bound can be improved to
\begin{equation}
    \max_{\rho\in\mathcal{S}(\mathcal{H}_2^{\otimes n})}\ddmin(\rho)> n - 2\log n - \log\ln2+1+\epsilon
\end{equation}
for any $\epsilon>0$, for sufficiently large $n$. (For any $\epsilon>0$, there exists some $N\in\mathbb{N}$ such that for any $n\geq N$, the above bound holds.)}
\end{thm}
\begin{proof}
Following the derivation of (\ref{eq:dminbound}) in the proof of Theorem~\ref{thm:haar}, we obtain
\begin{equation}
    \Pr\left\{\ddmin(\psi) \leq \gamma  \right\} < 2^{c_n n^2}\exp(- 2^{n-\gamma}), 
\end{equation}
where $c_n = \log(|\stab_n|)/n^2$.  Therefore, as long as
\begin{equation}
    \gamma < n-2\log n-\log(c_n\ln 2),  \label{eq:gamma}
\end{equation}
it holds that $\Pr\left\{\ddmin(\psi) \leq \gamma  \right\} < 1$ and thus $\max_\psi\ddmin(\psi)>\gamma$.   For $n\geq 7$, it holds that $c_n<0.74$, and thus $\max_\psi\ddmin(\psi)> n-2\log n -\log(0.74\ln 2) > n-2\log n + 0.96$.   Recall that $\ddmin(G^{\otimes n}) = \log(3-\sqrt{3}) n \gtrsim 0.34 n$, where $G = \frac{1}{2}(I+\frac{X+Y+Z}{\sqrt{3}})$.  For $n<7$, it can be verified that $n-2\log n + 0.96 < 0.34 n$ holds. So the first claimed bound follows.

To obtain the second bound for large $n$, recall that $c_n = 1/2+o(1)$  as $n\rightarrow\infty$ \cite{AaronsonGottesman04} and apply it to (\ref{eq:gamma}). Substituting this into (\ref{eq:gamma}) leads us to the claimed bound.
\end{proof}
\red{
\begin{rem}
The feature of $\stab$ used in the proofs of Theorems \ref{thm:haar} and \ref{thm:maxdmin} is the number of pure stabilizer states. In particular, the key message that $\ddmin > n - O(\log n)$ typically holds essentially comes from the number of free pure states being $2^{\mathrm{poly}(n)}$ and can thus be generalized to all theories with this property.
\end{rem}
}




In conclusion, we see that the maximum values of all consistent magic measures \red{($f$ such that $\ddmin\leq f\leq LR$)} are approximately $n$ for $n$-qubit states.
The results \red{potentially have} implications on the \red{asymptotic} reversibility of magic 
state transformations. We say a theory is reversible if resource states can be 
transformed back and forth using the free operations without loss. 
\red{The main result of Ref.~\cite{BrandaoGour15} is} that reversibility holds asymptotically, i.e.,
in the i.i.d.~limit \red{and with respect to the transformation rate}, for general resource theories satisfying several natural 
axioms, if the set of free operations not only includes all resource non-generating 
operations, but one allows \red{a certain class of} approximately resource non-generating operations, 
which is a bit unsatisfying from a fundamental conceptual point of view. \red{Indeed, it is recently shown \cite{LamiRegula21} that entanglement theory with exactly resource non-generating operations is asymptotically irreversible. On the other hand, for example, the resource theory of coherence under the non-generating set  MIO is already asymptotically reversible \cite{Using+reusing}. 
For magic, our Theorems \ref{thm:LR-max} and \ref{thm:maxdmin} imply the existence of a sequence of asymptotic golden currency states.
By the results in Ref.~\cite{LiuBuTakagi19:oneshot}, any resource theory such that $LR(\rho_n)$ and $\ddmin(\rho_n)$ are asymptotically equal to $n$ for some sequence
of $n$-qubit ``currency'' states $\rho_n$, is asymptotically reversible if and only if the two rates of distillation and of formation, to and from the currency states, respectively, are asymptotically equal.}

Now we \red{comment on the implications to the classical simulation of quantum computation.} 
\red{An idea that has drawn considerable interest is to devise  classical simulation algorithms based on the efficient simulation of stabilizer quantum computation \cite{Gottesman97:thesis,AaronsonGottesman04}.  
Here, one considers the general quantum computation model built upon Clifford operations aided by magic states, which are used, e.g.,~to emulate non-Clifford gates by state injection \cite{GottesmanChuang} or as the resource state of Pauli MBQC (see Sec.~\ref{sec:mbqc}). 
Since the Clifford part is ``cheap,'' we  are particularly interested in
 how the cost of the simulation algorithms depends on magic and how to optimize it.
}
For example, there are two leading methods: (i) Stabilizer decomposition, for which 
the cost depends on the (smooth) stabilizer rank \cite{BravyiSmithSmolin16:stabrank,BravyiGosset,Bravyi2019simulationofquantum}; 
(ii) Quasiprobability method based on stabilizer pseudomixture (also applies to mixed states), for which the cost depends on the free robustness of magic \cite{HowardCampbell17}.  
\red{Given the belief that the cost scales at least exponentially on magic (otherwise, there will be improbable complexity theory consequences), the efforts are devoted to reducing the exponent. Improvements over the worst-case, brute-force simulation cost rely on input magic states with special structures, such as a \red{collection} of $\ket{T}$ states that admit low-rank stabilizer decompositions \cite{BravyiSmithSmolin16:stabrank,BravyiGosset,Bravyi2019simulationofquantum}.}
The fact that almost all states must have $LR \approx n$ and maximal stabilizer rank (because lower-rank states are only a finite number of lower-dimensional manifolds, 
which form a measure-zero set)  tells us that these simulation methods typically give us 
no improvement, \red{even in the exponent}, over brute-force simulation. \red{That is, to facilitate even slight quantum advantage, let alone a significant one, the resource magic states have to have very special structures.}. 
\red{Interestingly, on the other hand, it turns out the typical magic results also indicate that most states are not able to supply nontrivial advantages over classical methods in solving $\mathsf{NP}$ problems in the Pauli MBQC model, despite being difficult to simulate using classical methods. 
This will be unraveled in Sec.~\ref{sec:mbqc}. }


\section{Hypergraph states \protect\\ and Boolean functions}
\label{sec:hyper}
\red{The preceding analysis shows that almost all quantum states have close to maximal magic, with respect to any magic measure.}
But we lack explicit constructions for highly magical states. 
Here we go in this direction by looking at hypergraph states, which are generalizations of graph states that possess highly flexible entanglement structures determined by an underlying hypergraph. 

We first formally define graph  and hypergraph states.
Graph states constitute an important family of many-body quantum 
states that plays key roles in various areas of quantum information, 
in particular, quantum error correction and MBQC. 
Given a graph $G=\{V,E\}$ defined by a set of $n$ vertices $V$ and a set 
of edges $E$, the corresponding $n$-qubit graph state is given by
\begin{equation}
  \ket{\Psi_G} := \prod_{\substack{i_1,i_2\in V\\ \{i_1,i_2\}\in E}} 
                                          CZ_{i_1 i_2}H^{\otimes n}\ket{0}^{\otimes n},
\end{equation}
\red{where $CZ$ and $H$ are respectively the standard controlled-$Z$ and Hadamard gates}.
Note that both gates are in the Clifford group, so graph states are stabilizer states. Conversely, it is known that every stabilizer state is equivalent to a graph state, up to a tensor product of local Clifford unitaries \red{\cite{Schlingemann,hein2006entanglement}}. Graph states thus already exhibit rich entanglement structures, indeed the same as general stabilizer states, which include most quantum error correcting codes known.
More generally, one can define \emph{hypergraph states} \cite{Rossi13:hypergraph} based on hypergraphs, where the hyperedges may contain $k\geq 2$ vertices and represent 
$C^{k-1}Z$ gates that acts $Z$ on one of the qubits conditioned on the $k-1$ others being 1. 
That is, given a hypergraph $\widetilde{G}=\{V,E\}$ defined by a set of $n$ vertices $V$ and a set of hyperedges $E$, the corresponding $n$-qubit hypergraph state is given by
\begin{equation}
  \ket{\Psi_{\widetilde{G}}} := \prod_{\substack{i_1,\cdots,i_k\in V\\ \{i_1,\cdots,i_k\}\in E}} 
                                 C^{k-1}Z_{i_1 \cdots i_k}H^{\otimes n}\ket{0}^{\otimes n}.
\end{equation}
\red{It is important to note that the $C^{k-1}Z$ gates when $k>2$, that are additionally allowed compared to the graph states, are not Clifford gates, and thus may generate magic. (More precisely, $C^{k-1}Z$ gates are in the $k$-th level but not the $(k-1)$-th level of the Clifford hierarchy \cite{GottesmanChuang,CuiGottesmanKrishna}.)  Because of the rich structure of the $C^{k-1}Z$ gates, the hypergraph states provide us with a natural, flexible family of many-body magic states.}

An important observation is that the hypergraph (including graph) states admit representations in terms of Boolean functions:
\begin{equation}
  \ket{\Psi} = 2^{-n/2}\sum_{x\in \mathbb{Z}_2^n}(-1)^{{f}(x)}\ket{x}.
\end{equation}
Here $f(x):\mathbb{Z}_2^n\rightarrow \mathbb{Z}_2$ is a Boolean function
\begin{equation}
    f(x) = \sum_{\substack{i_1,\cdots,i_k\in V\\ \{i_1,\cdots,i_k\}\in E}} 
                                  x_{i_1} \cdots x_{i_k},
\end{equation}
which we call the characteristic function of the hypergraph state $\ket{\Psi}$.
Each $f$ corresponds to a hypergraph state (modulo a global phase) and 
there are $2^{2^n-1}$ possibilities \cite{Rossi13:hypergraph}.
For a graph state,
$
  f(x) = \sum_{\substack{i_1,i_2\in V\\ \{i_1,i_2\}\in E}} x_{i_1} x_{i_2} 
$
is a function with only quadratic terms, where each term corresponds to an edge, 
and there are now only $2^{n \choose 2}$ possibilities.  
Any quadratic characteristic function (which may additionally include linear 
terms $x_i$) induces a stabilizer state because a term $x_i$ simply corresponds 
to a Pauli-$Z$ gate on the $i$-th qubit. 
Call such states induced by quadratic characteristic functions \emph{quadratic states} and denote the set of quadratic states $Q$. \red{Quadratic states are graph states with additional phases.}  Note that although the set of quadratic states does not include all stabilizer states (with additional local $H$ and $P$ freedom; see Eq.~(\ref{eq:HP})), that is, $Q\subsetneq \stab$, the size of $Q$ is close to that of $\stab$, both scaling roughly as $2^{n^2/2}$ asymptotically.

This formalism allows us to analyze certain magic properties of hypergraph 
states through Boolean functions.  
Given two hypergraph states $\ket{\Psi}$ and $\ket{\Psi'}$ with 
characteristic functions $f$ and $f'$ respectively, we have
\begin{equation}
  \bracket{\Psi}{\Psi'} 
      = 2^{-n}\sum_{x\in \mathbb{Z}_2^n}(-1)^{f(x)+{f}'(x)} 
      = 1-2^{1-n}\mathrm{wt}(f+f'),
  \label{eq:wt}
\end{equation}
where $\mathrm{wt}(f)$ denotes the Hamming weight of $f$, i.e.,~the number of 1's in the 
truth table of $f$. Therefore, $\mathrm{wt}(f+f')$ (also called the Hamming distance 
between $f$ and $f'$) essentially counts the number of non-collisions between $f$ and $f'$.
Here, we are interested in the minimization of $\mathrm{wt}(f+f')$ over all quadratic 
$f'$, 
namely the \emph{second-order nonlinearity} or \emph{nonquadraticity} 
of $f$, formally defined as
\begin{equation}
    \chi(f) := 
    \red{\min_{\text{quadratic}~f'}}
    \mathrm{wt}(f+f'). 
\end{equation}
\red{
Note that the codewords of the $r$-th order binary Reed--Muller codes of length $2^n$, denoted by $RM(r,n)$, are given by Boolean functions of algebraic  degree at most $r$ on  $n$ variables \cite{carlet_2010}.  That is, quadratic functions generate $RM(2,n)$.  
}
The nonquadraticity generalizes the well-studied nonlinearity of Boolean functions 
(an associated key notion is that of a ``bent function''), which has important applications in cryptography and 
coding theory \red{(see, e.g.,~Refs.~\cite{carlet_2010,Carlet2011})}.   
This leads to lower bounds on the maximum overlap between $\ket{\Psi}$ with stabilizer states,
because a quadratic characteristic function induces a stabilizer state as argued.
Using Eq.~(\ref{eq:wt}), we obtain the following bound on $\ddmin$ in terms of nonquadraticity:
\begin{equation}
    \ddmin(\Psi) \leq -\log \max_{q\in Q} |\bracket{\Psi}{q}|^2 
     = -2\log(1-2^{1-n}\chi(f)).
     \label{eq:ddmin_chi}
\end{equation}
\red{As mentioned, it is known that all stabilizer states can be generated by single-qubit Clifford gates acting on graph states (which form a subset of $Q$)  \cite{Schlingemann,hein2006entanglement}}.  Note that $Q$ is closed under single-qubit Pauli operators up to global phases.  So any pure stabilizer state $\ket{s}$ takes the form
\begin{equation}
    \ket{s} = \bigotimes_{i\in\mathcal{I}}P_i \bigotimes_{j\in\mathcal{J}}H_j \ket{q},
    \label{eq:HP}
\end{equation}
where $P,H$ are respectively the phase gate and the Hadamard gate, $\ket{q}\in Q$ is a quadratic state, and $\mathcal{I},\mathcal{J}$ are respectively the set of indices of qubits that $P,H$ act on.    
\red{So, the interesting question of how tight the inequality in (\ref{eq:ddmin_chi}) is, namely how close the maximum overlap  with respect to $Q$ is compared to that with respect to $\stab$ (note that $Q$ is almost as large as $\stab$), comes down to the effects of such local $P,H$ gates. 
We leave this as an open problem for future study.}

The above technique allows us to obtain bounds on the magic of generally highly entangled \red{states} of arbitrary size, from the nonquadraticity of Boolean functions.  In Sec.~\ref{sec:spt}, we shall use this technique to analyze certain physically motivated hypergraph states, which serve as explicit examples.
Also, results on Boolean functions and Reed--Muller codes lead to several general understandings.
\red{The maximum possible nonquadraticity $\chi$ is equivalent to the \emph{covering radius} of the 
second-order Reed--Muller code $RM(2,n)$, denoted by $r(RM(2,n))$ \cite{Cohen:covering-codes}}.
Determining the covering radii for codes is an important but generally 
difficult task. For general $n$, there are only bounds known for $r(RM(2,n))$.   
The best upper bound to our knowledge is from Ref.~\cite{CarletMesnager07:coveringradii}, 
\begin{equation}
   \max \chi(f) \equiv r(RM(2,n)) \leq 2^{n-1} - \frac{\sqrt{15}}{2}2^{n/2}+O(1).
\end{equation}
Thus, by (\ref{eq:ddmin_chi}),
\begin{equation}
    \ddmin(\Psi) \leq n-\log 15 + o(1).
\end{equation}
We learn from this bound that for any hypergraph state $\ket{\Psi}$, the 
min-relative entropy of magic $\ddmin(\Psi)$ is upper bounded by $n-\log 15 \approx n - 3.9$ 
in the large-$n$ limit.  
We also have lower bounds coming from simple covering arguments \cite{COHEN1992147}:
\begin{equation}
  \max \chi(f)  \geq 2^{n-1}-\frac{\sqrt{\ln 2}}{2} n 2^{n/2} + O(1).
\end{equation}
This implies that there exists a hypergraph state $\ket{\Psi}$ such that 
\begin{equation}
  -\log \max_{q\in Q} |\bracket{\Psi}{q}|^2 \geq n - 2\log n - \log\ln 2 + o(1). 
  \label{eq:hypergraph_max}
\end{equation}
Recall that $Q$ and $\stab$ are very similar sets, and that we believe that the left hand side is close to $\ddmin$ especially in the present high-magic regime. 
This bound is very close to the Haar-random value in Theorem \ref{thm:haar}, but unfortunately is also not constructive. 
To our knowledge, the Boolean function with the largest nonquadraticity in the literature is the modified Welch function (see, e.g.,~Ref.~\cite{Carlet08}) defined as
$
    f_{\rm W}(x) = {\rm tr}(x^{2^r+3})$ where $r=\frac{n+1}{2}, n~\text{odd}$.  We have $\chi(f_{\rm W})\approx 2^{n-1}-2^{(3n-1)/4}$, so, for the corresponding hypergraph state $\Psi_{\rm W}$, it holds that
    \begin{equation}
       \ddmin(\Psi_{\rm W}) \leq -\log \max_{q\in Q} |\bracket{\Psi_{\rm W}}{q}|^2 \approx 0.5n. 
    \end{equation}
\red{As a side note, the algebraic degree of $f_{\rm W}$ (and thus the largest size of the hyperedges of the hypergraph associated with $\Psi_{\rm W}$) is only 3.}
Although we know that the typical magic of a random state is close to $n$, we do not yet have specific constructions of many-body states with such high magic. 
The situation is reminiscent of, e.g.~the superadditivity of classical capacity 
\cite{Hastings09:superadditivity}, which is shown for certain random 
ensembles, but no deterministic construction is known.
\red{To conclude, the quantification of many-body magic provides a new, physical motivation for further studying the nonquadraticity of Boolean functions, especially high-degree ones.}

\section{Pauli Measurement based quantum computation}
\label{sec:mbqc}
MBQC \cite{RaussendorfBriegel01,RaussendorfBrowneBriegel03} is a profound and promising model for quantum computation, where one prepares a many-body entangled  state offline, and then executes the computation by a sequence of local measurements adaptively determined by a classical computer on this resource state.  

This model is naturally tied to resource theory as it essentially formalizes quantum computation as free online manipulations of a resource state. Standard MBQC only allows single-qubit measurements, so entanglement among qubits in the initial state becomes the key resource feature, and a core line of study is to understand the degree of entanglement that supports universal quantum computation (see, e.g.,~Refs.~\cite{VandenNest06,Nest_2007,GrossFlammiaEisert08,BremnerMoraWinter08}). 

Here we consider a variant where one is restricted to measuring mutually compatible Pauli observables (including multi-qubit ones such as $X\otimes X$, which cover entangled measurements, for the greatest generality), which we call Pauli MBQC.  This is desirable both practically and conceptually, in a similar spirit as the 
\red{well-known magic state model based on magic state distillation and injection \cite{BravyiKitaev,GottesmanChuang}}.  Clearly, we would like the online procedures to be as simple  as possible for implementation and fault tolerance;
moreover, the ``magic'' of the computation is now isolated to offline resource state preparation, which paves the way for understanding and analyzing the genuine ``quantumness'' in MBQC models.
\red{(In fact, the magic state model is a subclass of Pauli MBQC where no entanglement is required in the input resource state.)}
As a comparison, the standard MBQC using cluster states \cite{RaussendorfBriegel01,VandenNest06} requires online measurements in ``magical'' bases since cluster states are stabilizer states, leaving certain computational non-classicality in the online part. 

More generally, in this Pauli MBQC setting, it is clear that many-body magic states are necessary \red{to achieve quantum speedups} due to the Gottesman--Knill theorem.  Known resource states useful for Pauli MBQC include a certain hypergraph state introduced by Takeuchi, Morimae and Hayashi \cite{TakeuchiMorimaeHayashi19}, and the Miller--Miyake state \cite{MillerMiyake16} (which will be discussed in the next section).
A central question is whether all magic resource states can supply significant speedups over classical algorithms, or support universal quantum computation.   In the case of standard MBQC, it is known that resource states with ``too much'' entanglement (and thereby most states) are not useful for quantum speedups   \cite{GrossFlammiaEisert08,BremnerMoraWinter08}.  Here we show that similar rules also hold for Pauli MBQC and magic states, by adapting the arguments in Ref.~\cite{GrossFlammiaEisert08}.
\red{Intuitively, if the resource \red{state} is too ``magical,'' any Pauli measurement scheme will produce outcomes that are too uniform across all possible ones so that they can be well simulated by classical randomness, or more precisely, so that the set of witnesses for the problem is sufficiently large to allow for an efficient probabilistic search. Indeed, the known examples of Pauli MBQC such as Takeuchi--Morimae--Hayashi \cite{TakeuchiMorimaeHayashi19} and Miller--Miyake \cite{MillerMiyake16} are based on resource states prepared by Clifford+$CCZ$ circuits with specific structures, \red{which are expected to have only ``medium'' magic ($\sim cn$ where $0<c<1$)}; see Sec.~\ref{sec:spt}.  The formal result and proof go as follows.}
\begin{thm}
  Pauli MBQC with any $n$-qubit resource state $\ket{\Psi}$ with $\ddmin(\Psi)\geq n - O(\log n)$ cannot achieve superpolynomial speedups over $\mathsf{BPP}$ machines (classical randomized algorithms) for problems in $\mathsf{NP}$.
  \label{thm:mbqc}
\end{thm}
\begin{proof}
First note that all Pauli observables have eigenvalues $\pm 1$, and those defined nontrivially on multiple qubits (joint measurements) have degenerate eigenstates. 
Suppose that we measure $k$ (mutually compatible) observables, labeled by $P_i, i=1,...,k$. 
The measurement outcome of $P_i$ \red{is} a binary variable $y_i = \pm 1$, so the collective outcome can be represented by a bit string $y = y_1,...,y_k$ with $2^k$ possible values, each of which corresponds to a subspace of the entire Hilbert space.  
The probability of obtaining $y$ is given by
\begin{equation}
    p(y) = \Tr (\Pi_{y}\ket{\Psi}\bra{\Psi}),
\end{equation}
where $\Pi_{y}$ is the projector onto the subspace corresponding to $y$.  Note that $\Pi_y$ takes the form
\begin{equation}\label{eq:pi_y}
    \Pi_y = \sum_{j=1}^{2^{n-k}}\ketbra{s_{j,y}}{s_{j,y}},
\end{equation}
where $\{s_{j,y}:j=1,...,2^{n-k}\}$ is \red{an orthogonal basis of states that are} 
stabilized by $\{y_i P_i\}$, $i=1,...,k$. 
There are $2^{n-k}$ such stabilizer states because each measurement halves the dimension. \red{More concisely, measuring a set of mutually compatible Pauli observables is equivalent to measuring a partition of the identity composed of stabilizer codes of the form given by (\ref{eq:pi_y}).}  Therefore, we have
\begin{equation}
    p(y) = \sum_{j=1}^{2^{n-k}}|\bracket{s_j}{\Psi}|^2 \leq 2^{n-k-\ddmin(\Psi)},
\end{equation}
by using standard properties of the trace function and the definition of $\ddmin$.   Suppose that the algorithm succeeds with probability $\geq 2/3$. That is, let $G$ be the set of strings leading to valid solutions, then $\sum_{y\in G}p(y) \geq 2/3$.  Therefore, the size of $G$ obeys
\begin{equation}
    |G| \geq {2^{-n+k+\ddmin(\Psi)+1}}/{3}.
\end{equation}
As a result, one can simulate the above procedure by a classical randomized algorithm in polynomial time, namely in $\mathsf{BPP}$. \red{The more specific argument goes as follows. Let $N$ be the input size of the $\mathsf{NP}$ problem. Since the quantum computation is supposed to be efficient we have $n = \mathrm{poly}(N)$. One generates $k$ uniformly random bits from an i.i.d.\ source and feeds it into the polynomial-time verifier of the $\mathsf{NP}$ problem 
to see whether it succeeds (this checking step takes time at  most $\mathrm{poly}(N) = \mathrm{poly}(n)$).} If it fails, generate another random string and check again.  The probability that the algorithm still has not succeeded after $t$ repetitions satisfies
\begin{equation}
    p_{f}  =  (1-|G|/2^k)^t  \leq \left(1-\frac{2^{-n+\ddmin(\Psi)+1}}{3}\right)^t.
\end{equation}
So to achieve success probability $\geq 2/3$, namely $p_f \leq 1/3$, the number of repetitions needed satisfies
\begin{equation}
    t \leq 3\log 3\cdot 2^{n-\ddmin(\Psi) -1}.
\end{equation}
When $\ddmin(\Psi)\geq n - O(\log n)$, it can be directly seen that $t$ is upper bounded by $\mathrm{poly}(n)$. Multiplying the checking time, it can be concluded that the total runtime of this classical simulation is upper bounded by $\mathrm{poly}(n)$.
\end{proof}


Combining with results in Sec.~\ref{sec:bounds}, we see that almost all states are useless for Pauli MBQC in a strong sense:
\begin{cor}
 The fraction of states (w.r.t.~Haar measure) that can supply nontrivial quantum advantages via Pauli MBQC is exponentially small in $n$.
\end{cor}

\red{We conclude this section by remarking that, as with many good things, one can have too little and too much  magic to be of any good: indeed, the behavior under Clifford operations of states with ``too little'' magic can be efficiently simulated classically, ruling out any computational advantage; while ``too much'' magic means that the state may in general be hard to simulate, but its behavior in a Pauli MBQC protocol is trivial, i.e., essentially random, so in this context there is no quantum advantage either. This highlights that quantum computation requires very delicate structures or features of quantum systems---although most of them are hard to simulate classically or contain near-maximal quantum resource, most of them are also useless. Only in an intricate intermediate regime can they manifest a quantum computational advantage.}

\section{Quantum phases of matter}
\label{sec:spt}

\zwedit{ The Clifford group and stabilizer formalism have become standard notions and tools in recent studies of condensed matter physics, but so far there is little discussion on their physical relevance and the role of magic, especially at a quantitative level.
Here we would like to present some basic discussions and results on the magic of quantum many-body systems of interest from a phase of matter perspective, in the hope of stimulating further explorations in this direction. 
This section can also be viewed as a case study of the techniques introduced in Sec.~\ref{sec:hyper} for quantitatively analyzing many-body magic.   

Here we consider SPT phases, which have drawn great interest in the condensed matter community (see, e.g.,~Refs.~\cite{Senthil15,QImeetsQM} for introductions)  and in particular, been studied as a useful type of many-body resource states useful for MBQC (see, e.g.,~Ref.~\cite{Wei18:mbqc} for a review).  
\red{It has recently been realized that a wide range of nontrivial SPT phases in $\geq 2$D must contain magic that is ``robust'' in a physical sense \cite{EKLH20}},  indicating that  magic is a characteristic feature underpinning the physics of such systems.    
Here we showcase how to apply the Boolean function techniques introduced in Sec.~\ref{sec:hyper} to representative 2D SPT states. 
For example, it is known that the Levin--Gu \cite{LevinGu12} and Miller--Miyake \cite{MillerMiyake16} models have ground states that are hypergraph states prepared by Clifford+$CCZ$ circuits defined on corresponding lattices, so that the characteristic functions of these ground states are restricted to cubic ones, namely third-order Reed--Muller codes $RM(3,n)$.  
For concreteness, think about the well-known Levin--Gu state $\ket{\Psi_{\mathrm{LG}}}$ \cite{LevinGu12} defined on the 2D triangular lattice (see Fig.~\ref{fig:lattice}), which takes the form
\begin{equation}
  \ket{\Psi_{\mathrm{LG}}} = U_{{CCZ}} U_{{CZ}} U_{{Z}} H^{\otimes n}\ket{0}^{\otimes n},
  \label{eq:triangular}
\end{equation}
where $U_{CCZ}, U_{CZ}, U_{Z}$ are respectively composed of $CCZ$, $CZ$, $Z$ gates acting on all triangles, edges, and vertices. 
More generally, consider third-order hypergraph states
\begin{equation}
    \ket{\hat\Psi} = U_{{CCZ}} \ket{\Phi},\quad \ket{\Phi}\in Q,  \label{eq:cczstate}
\end{equation}
defined on 2D triangulated lattices (such as the ordinary triangular lattice and the Union Jack lattice, as depicted in Fig.~\ref{fig:lattice}), where $U_{CCZ}$ represents $CCZ$ gates acting on all triangles. 
\begin{figure}[t]
 \includegraphics[width=\columnwidth]{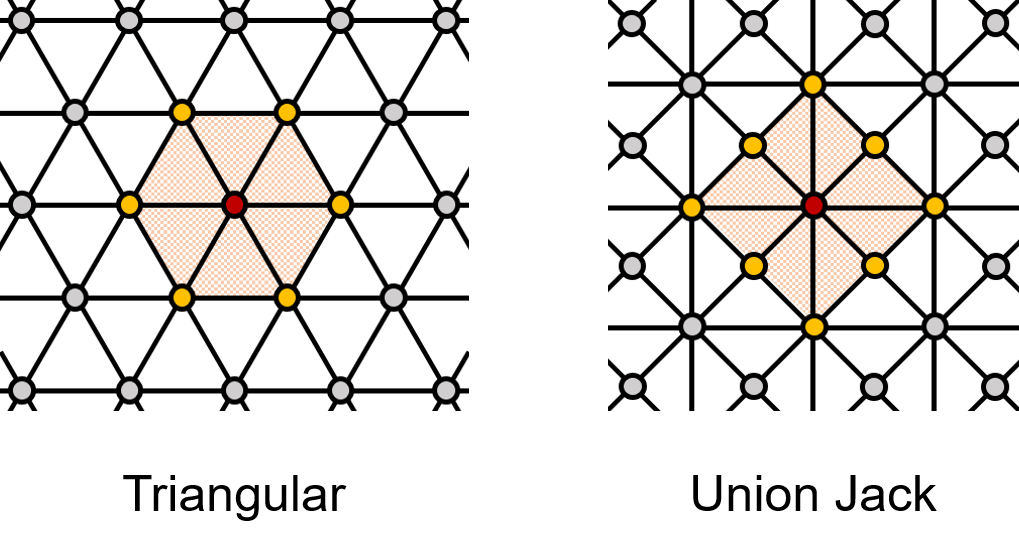}
 \caption{2D triangulated lattices.  The shaded area represents a unit cell, based on which we decompose the underlying Boolean functions of the systems and derive bounds on their nonquadracity (details in Appendix~\ref{app:ccz}).  \label{fig:lattice}}
\end{figure}
Note that the Clifford+$CCZ$ preparation circuits of such states are in the third level of the Clifford hierarchy \cite{GottesmanChuang}. 
Such states are called ``Clifford magic states'' in Ref.~\cite{Bravyi2019simulationofquantum} 
and are shown to have the property that the ``stabilizer extent,'' 
$\xi(\Psi):=\min\|c\|_1^2$ where $c$ is the amplitude vector of a decomposition 
into pure stabilizer states, is equal to $2^{\ddmin(\Psi)}$ due to convex duality.
It is known that the logarithm of the stabilizer extent $\log\xi$ and max-relative entropy 
monotone $\ddmax$ (and thus also generalized robustness) are 
equivalent \cite{Regula2017,LiuBuTakagi19:oneshot}.  Therefore we have the collapse property $\ddmax(\hat\Psi)=\ddmin(\hat\Psi)$.  Using  techniques from Refs.~\cite{Cubic_book,Kolokotronis09:quadratic}, we rigorously prove the following crude bounds for the two example lattices (which hold for both open and periodic boundary conditions):
\begin{itemize}
    \item Triangular lattice: $\ddmax(\hat\Psi)=\ddmin(\hat\Psi) < 0.56n.$
\item Union Jack lattice: $\ddmax(\hat\Psi)=\ddmin(\hat\Psi) < 0.46n.$
\end{itemize}
Roughly speaking, our approach is to find proper decompositions of the cubic characteristic functions based on cell structures of the underlying lattices (as illustrated in Fig.~\ref{fig:lattice}) that allow us to bound its distance from certain quadratic functions (and thus the nonquadraticity). See  Appendix~\ref{app:ccz} for technical details of the derivation.
Note that we expect the above bounds to be loose, and it can likely be shown that $\ddmax(\hat\Psi)=\ddmin(\hat\Psi) \leq \left(2-\frac{2}{3}\log 6\right)n \lesssim 0.28n$ for all regular triangulated lattices (also see Appendix~\ref{app:ccz} for more detailed discussions and probable ways to improve the bounds), which is achieved by disjoint $CCZ$ gates (namely, $CCZ^{\otimes \frac{n}{3}}$) because $\ddmax(CCZ|{+++}\rangle)/3 = \ddmin(CCZ|{+++}\rangle)/3 = \log(16/9)/3 = 2-\frac{2}{3}\log 6$ \cite{Bravyi2019simulationofquantum}.  
Also note that the maximum product-state value is $\log(3-\sqrt{3})n\approx 0.34n$, achieved by the product of qubit golden state  $\frac{1}{2}(I+\frac{X+Y+Z}{\sqrt{3}})$. 
\red{So an observation is that, although the $CCZ$ gates can generate rich entanglement structures that supply interesting topological properties, the many-body magic of the corresponding SPT states is rather weak (likely not even higher than certain states with no entanglement), despite being generically necessary and robust \cite{EKLH20}.} This makes the role of magic more curious.
Note that, e.g.~the fixed point of Miller--Miyake model on the Union Jack lattice (which satisfies (\ref{eq:cczstate})) is known to be universal for Pauli MBQC \cite{PhysRevLett.120.170503}, so the bound is consistent with Theorem~\ref{thm:mbqc}. 
Nevertheless, note also that the magic of such many-body states is still an extensive quantity, i.e.,~scales with the system size $n$.  A simple argument is that one can do Pauli measurements on vertices in a periodic manner (Fig.~\ref{fig:extensive} illustrates the case of the Union Jack lattice; the idea can be generalized to other regular lattices), which leaves a periodic array of $O(n)$ uncoupled $CCZ$ blocks, each containing a certain amount of magic. 
\begin{figure}[t]
 \includegraphics[width=0.42\columnwidth]{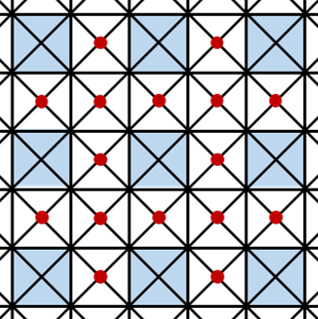}
 \caption{Extensiveness of magic. After measuring the red vertices by Pauli observables, the system is left with decoupled $CCZ$ blocks (colored in blue). \label{fig:extensive}}
\end{figure}
Note that a  characteristic feature of SPT phases is that they are short-range entangled, which accords with the rather weak magic.  It also indicates that for SPT phases, the method of calculating the magic of small lattices  and then ``scaling up''  the results may help approximate the magic of the whole system well.
For future  work, it would be particularly interesting to look into long-range  entangled, intrinsically topologically ordered systems like topological codes. 


We anticipate that the study of many-body magic will provide a new and useful perspective on characterizing and classifying quantum phases of matter. Since the family  of hypergraph states can describe very rich many-body entanglement structures that underlie the interesting physics of quantum matter, we expect the Boolean function techniques just introduced to be widely useful.
A natural direction is to further explore the connections between magic and computational complexity or power of phases of matter.  For example, a direct question following the above discussions is whether magic can be used to diagnose whether a phase is universal for Pauli MBQC, or more generally certain notions of ``quantum computational phase transitions.''  In particular, noting that the above studied Miller--Miyake and Levin--Gu models are known to be universal on the Union Jack lattice but likely not universal on the triangular lattice \cite{MillerMiyake16,Wei18:mbqc}, \red{it would be interesting to further understand what kinds of magic properties (e.g.,~scaling factors, topological and locality features) really determine the computational power.}  
On the other hand, magic determines the cost of many standard methods for preparing and simulating the systems and could plausibly be connected to related problems like the notorious sign problem in various forms (see, e.g.,~Refs.~\cite{gubernatis_kawashima_werner_2016,DBLP:journals/qic/BravyiDOT08,Hastings16:signproblem}).  For example, the extensive property directly indicates that the run time of the quasi-probability sampling algorithm of Howard and Campbell \cite{HowardCampbell17} is exponential.  
\newref{Also, as recently found in  Refs.~\cite{Sarkar_2020,WhiteCaoSwingle}, the behaviors of many-body magic have strong relevance to the phase transitions of certain important physical systems, indicating that magic could be a very useful diagnostic in many-body physics.}
We finally refer interested readers to Ref.~\cite{EKLH20}, which shows the necessity and robustness of magic throughout certain types of $\geq 2$D  SPT phases, and contains more results and discussions about magic from condensed matter perspectives, in relation to symmetries, sign problems, MBQC, and more. 


}

\section{Concluding remarks}
In this work, we formally studied the magic of many-body entangled quantum systems from multiple aspects, and proposed it as a potentially useful probe of many-body quantum complexity.
We found that magic is a highly nontrivial resource theory with complicated mathematical structures, so that the calculation and analysis of many-body magic measures are in general  difficult but very interesting.
Our results indicate an intriguing interplay between magic and entanglement worth further study: although magic and entanglement are disparate notions, they may be correlated in the highly entangled regime.  For example, we now know that quantum states are typically almost maximally entangled and magical at the same time, but do highly magical states have to be highly entangled in some sense, or vice versa?
On a related note, the problem of explicitly constructing scalable families of nearly maximally magical states with respect to any of the measures we investigated, is still wide open.


As is often the case in resource theories, some of the most interesting quantifiers are hard to compute, and even their upper and lower bounds may present serious computational challenges. In the case of magic, the complexity of calculations  scales badly with $n$ because of the exponential growth of the number of stabilizer states (despite the observation that the free and generalized robustnesses are \red{respectively linear and semidefinite} programs).
The search for easier bounds thus remains highly important. For certain condensed matter systems it may be sufficient to calculate values for small lattices, but the general case remains to be explored.

With the present work, we hope to raise further interest in  magic in entangled quantum systems, and magic as a new approach to many-body physics. Indeed, 
as discussed in the paper, many-body magic could be very relevant to the characterization of quantum complexity of phases of matter, such as the cost of simulating certain phases and the computational power of them.


\begin{acknowledgments}
We thank Tyler Ellison, David Gosset, Daniel Gottesman, Tim Hsieh, Linghang Kong, Kohdai Kuroiwa for discussions.
We also thank an anonymous referee for suggesting the generalization 
and alternative proof of Theorem~\ref{thm:LR-max}, and for permission to include 
them in the present paper. 
ZWL is supported by Perimeter Institute for Theoretical Physics.
Research at Perimeter Institute is supported in part by the Government of Canada through the Department of Innovation, Science and Economic Development Canada and by the Province of Ontario through the Ministry of Colleges and Universities. 
AW acknowledges financial support 
by the Spanish MINECO (projects FIS2016-86681-P
and PID2019-107609GB-I00/AEI/10.13039/501100011033) 
with the support of FEDER funds, 
and the Generalitat de Catalunya (project 2017-SGR-1127).
\end{acknowledgments}

\appendix

\section{Range of consistent resource measures}
\label{sec:range}
The following argument goes some way to justify that the min-relative entropy 
and the free robustness in some sense define a range  of consistent resource measures.  

Consider a theory with finite free robustness (i.e.,~satisfying condition 
FFR defined in Ref.~\cite{LiuBuTakagi19:oneshot}), such as magic theory in which we are 
interested here.  
By the definition of free robustness, there exists a free state 
$\delta\in\mathcal{F}$ such that 
$\frac{1}{1+R\left(\phi\right)} \phi+\frac{R\left(\phi\right)}{1+R\left(\phi\right)} \delta \in \mathcal{F}$.  
Consider the following cptp map
\begin{equation}
    \mathcal{E}(\omega)=(\tr\psi \omega) \phi+(1-\tr\psi \omega) \delta,
\end{equation}
where $\psi$ and $\phi$ are pure states.  Notice that $ \mathcal{E}(\psi) = \phi$.    It can be verified that if $\frac{1-\tr\psi \omega}{\tr\psi \omega} \geq R\left(\phi\right)$, that is, 
\begin{equation}
\tr\psi \omega \leq 2^{-LR(\phi)}, \label{eq:tr_lr}
\end{equation}
then $\mathcal{E}(\omega)\in\mathcal{F}$.
When $\ddmin(\psi) \geq LR(\phi)$, then, for any $\omega\in\mathcal{F}$, (\ref{eq:tr_lr}) holds and thus $\mathcal{E}(\omega)\in\mathcal{F}$.  That is, $\mathcal{E}$ is a resource non-generating operation.   To summarize, the condition $\ddmin(\psi) \geq LR(\phi)$ implies that there must exist a resource non-generating operation that accomplishes the one-shot transformation $\psi\rightarrow\phi$.

The consistency argument goes as follows.  Pick a standard reference resource measure $f_r$ such as the relative entropy of resource (or others satisfying $\ddmin\leq f_r \leq LR$).  It plays the role of fixing a standard normalization in order to, e.g.,~avoid ambiguities about constants.    Based on this, consider the following consistency condition of resource measure $f$: If the transformation $\psi\rightarrow\phi$ by free operations is possible, then 
\begin{align}
    f(\psi)\geq f_r(\phi)\quad \text{and}\quad
    f_r(\psi) \geq f(\phi).
\end{align}
If these are not satisfied then $f$ may be regarded inconsistent with the reference measure $f_r$ since comparisons with $f_r$ will rule out feasible free transformations. 

This consistency condition implies that $\ddmin\leq f \leq LR$.  Suppose that there exists $\psi$ such that $f(\psi) < \ddmin(\psi)$.  
Then there must exist some $\phi$ such that $f(\psi)< LR(\phi) \leq \ddmin(\psi)$, so $\psi\rightarrow\phi$ is feasible but $f(\psi) < f_r(\phi)$, violating the first consistency condition.  Similarly, suppose that there exists $\psi'$ such that $f(\psi') > LR(\psi')$.  Then there must exist some $\phi'$ such that $f(\psi') > \ddmin(\phi') \geq LR(\psi')$, so $\phi'\rightarrow\psi'$ is feasible but $f_r(\phi') < f(\psi')$, violating the second consistency condition.

We also refer readers to Refs.~\cite{KuroiwaYamasaki20,GourTomamichel} for other arguments about ranges of resource measures.

\section{Wigner negativity and robustness}
\label{sec:negativity}
Here we discuss the relations between computable magic measures based on Wigner negativity 
and free robustness, which have not explicitly appeared in the literature before.

Consider odd prime power dimension $D = d^n$, for which the stabilizer 
formalism in terms of the discrete Wigner function is well defined.  
Detailed introductions can be found in the literature; see, e.g.,~Ref.~\cite{Veitch14:magic_rt}.
Here we are interested in the widely used magic measures defined as follows.
Define the phase space point operators as 
\begin{equation}
  A_{\boldsymbol{0}}=\frac{1}{d^{n}} \sum_{\boldsymbol{u}} T_{\boldsymbol{u}}, 
  \text{ and } A_{\boldsymbol{u}}=T_{\boldsymbol{u}} A_{\boldsymbol{0}} T_{\boldsymbol{u}}^{\dagger},
\end{equation}
where $T_{\boldsymbol{u}}$ are the discrete Heisenberg-Weyl (generalized Pauli) operators:
\begin{equation}
  T_{\boldsymbol{u}} = \omega^{-\frac{a_1 a_2}{2}} Z^{a_1} X^{a_2}, \quad 
  \boldsymbol{u}=(a_1,a_2)\in\mathbb{Z}_d\times \mathbb{Z}_d
\end{equation}
where $\omega = e^{2\pi i/d}$. For composite systems,
\begin{equation}
  T_{\left(a_{1},a_{2}\right) \oplus\left(b_{1},b_{2}\right) \oplus \cdots \oplus \left(u_{1},u_{2}\right)} 
  = T_{\left(a_{1},a_{2}\right)} \otimes T_{\left(b_{1},b_{2}\right)} 
                                                       \otimes \cdots \otimes T_{\left(u_{1}, u_{2}\right)}.
\end{equation}
Then for state $\rho$ in dimension $D$, the corresponding discrete Wigner 
quasi-probability representation is given by
\begin{equation}
  W_\rho(\boldsymbol{u}) := \frac{1}{d^n} \tr A_{\boldsymbol{u}} \rho.
\end{equation}
It is known that pure stabilizer states are precisely those pure states that have
non-negative Wigner functions (discrete Hudson's theorem) \cite{Gross:Hudson}.
This motivates us to use the negative values of the Wigner functions to measure magic.
Define $\mathcal{W}$ to be the set of all real-valued functions $v$ on phase
space points $\boldsymbol{u}$ with the normalization 
$\sum_{\boldsymbol{u}} v(\boldsymbol{u}) = 1$. This is an affine space containing 
all Wigner functions $W_\rho$. In this space, identify the convex cone of non-negative functions, 
\begin{equation}
 \mathcal{W}_+ := \{ v\in\mathcal{W} : \forall\boldsymbol{u}, v(\boldsymbol{u}) \geq 0\},
\end{equation}
which, by definition, contains all non-negative Wigner functions, in particular those $W_\sigma$ of mixtures of stabilizer states $\sigma\in\stab$.

\begin{defn}[Sum negativity]
  The \emph{sum negativity} (or Wigner negativity) of a state $\rho$ is defined as 
  \begin{align}
    \mathscr{N}(\rho) 
          &:= \sum_{\boldsymbol{u}: W_\rho(\boldsymbol{u})<0 }\left|W_\rho(\boldsymbol{u})\right| \\
          &=  \frac{1}{2}\left(\sum_{\boldsymbol{u}}\left|W_\rho(\boldsymbol{u})\right|-1\right). 
  \end{align}
\end{defn}

\begin{thm}
  \label{thm:nr}
  For all states $\rho$,
  $\mathscr{N}(\rho) \leq R(\rho)$.
\end{thm}
\begin{proof}
The definition of the Wigner negativity is evidently equivalent to
\[
  \mathscr{N}(\rho) = \min s \text{ s.t. } W_\rho = (1+s)v - sv',\ v,v' \in \mathcal{W}_+.
\]
The optimal functions are $v \propto (W_\rho)_+ = \max\{W_\rho,0\}$
and $v' \propto (W_\rho)_- = \min\{W_\rho,0\}$. 

On the other hand, the definition of the free robustness is equivalent to
\[
  R(\rho) = \min s \text{ s.t. } W_\rho = (1+s)W_\sigma - sW_{\sigma'},\ \sigma,\sigma'\in\stab,
\]
because the Wigner function is an isomorphism, so the above condition is 
just one way to express $\rho = (1+s)\sigma - s\sigma'$.
Since all $W_\sigma,W_{\sigma'} \in \mathcal{W}_+$, the former minimization 
is a relaxation of the latter, and hence the claim follows. 
\end{proof}

Ref.~\cite{Veitch14:magic_rt} also defined an additive version of the sum negativity:
\begin{defn}[Mana]
The \emph{mana} of state $\rho$ is defined as 
$\mathscr{M}(\rho) := \log \left(\sum_{\boldsymbol{u}}\left|W_{\rho}(\boldsymbol{u})\right|\right)
                    = \log(2\mathscr{N}(\rho)+1)$.
\end{defn}

\begin{cor}
  For all states $\rho$, $\mathscr{M}(\rho)<LR(\rho)+1$.
\end{cor}
\begin{proof}
Elementary calculation yields:
\begin{align*}
  \mathscr{M}(\rho) &=    \log (2 \mathscr{N}(\rho)+1) \\ 
                    &\leq \log (2 R(\rho)+1)           \\ 
                    &<    \log (2 R(\rho)+2) = LR(\rho)+1,
\end{align*}
by Theorem \ref{thm:nr}.
\end{proof}

Therefore, each upper bound on $LR$ also give bounds on mana.

\zwnew{
This positive Wigner simplex contains the stabilizer polytope and is 
geometrically simpler, so the associated measures are easier to compute. 
Note that Theorem 8.4 of Ref.~\cite{GrossNezamiWalter17} gives a lower bound of $\mathscr{M}$ in terms of $\ddmin$ that works for small $\mathscr{M}$.   In general, can we lower bound $\mathscr{M}$ in terms of other smaller measures associated with STAB?
}

Also note that the negativity and free robustness measures are related to the 
runtimes or other costs of quasiprobability methods of classical simulation \cite{PashayanWallmanBartlett15:negativity,HowardCampbell17}.


\section{Clifford+$CCZ$ circuits, \protect\\ cubic Boolean functions, and SPT phases}
\label{app:ccz}

Here we provide technical details and extensive discussions on how to employ coding theory tools to analyze the magic of Clifford+$CCZ$ circuits of certain structures which correspond to 2D SPT phases of interest, supplementing Sec.~\ref{sec:spt}.

\red{Below we simply denote by $RM(r,n)$ the set of degree-$r$ Boolean functions on $n$ binary variables.} Let $f\in RM(3,n)$ be a cubic Boolean function with $n$ variables.
Suppose that there is a set of indices $\mathcal{X}=\{x_{c_1},\cdots,x_{c_s}\}$ such that every cubic term
involves one variable from $\mathcal{X}$, that is, $f$ takes the following form
\begin{equation}
    f(x) = \sum_{i=1}^s x_{c_i}q_i + q, \label{eq:class}
\end{equation}
where $q_i\in RM(2,n-s)$ is the quadratic function associated with $x_{c_i}$ so that $\sum_{i\in\{1,\cdots,s\}} x_{c_i}q_i$ is the cubic part of $f$, and $q$ is quadratic (containing linear terms as well).  We call (\ref{eq:class}) an order-$s$ decomposition of $f$.  Given some quadratic function $q = xQx^T$, let $2h_{q}$ be the rank of the \red{symmetric} matrix $Q+Q^T$.
Following the arguments in Section 4 of Ref.~\cite{Cubic_book}, we see that there exists a quadratic function $\tilde{q}$ (given in Theorem 3) such that 
\begin{align}
   \mathrm{wt}(f+\tilde{q}) &= 2^{n-1} - \sum_{r\in\mathbb{F}_2^s} 2^{n-s-1-h_{\sum_{i=1}^s r_i q_i}} \label{eq:ineq1}\\
   &\leq 2^{n-1} - \sum_{r\in\mathbb{F}_2^s} 2^{n-s-1-\sum_{i=1}^s r_i h_{q_i}} \label{eq:ineq} \\
   &= 2^{n-1}-2^{n-s-1} \prod_{i=1}^{s}\left(1+2^{-h_{q_{i}}}\right).  \label{eq:lattice_dmin}
\end{align}
Therefore, given an order-$s$ decomposition of $f$,  for its nonquadraticity, we have the following bound
\begin{equation}
    \chi(f) \leq 2^{n-1}-2^{n-s-1} \prod_{i=1}^{s}\left(1+2^{-h_{q_{i}}}\right),
\end{equation}
and thus for the corresponding state $\ket{\Psi_f}$, we have
\begin{align}
   \ddmax(\Psi_f) = \ddmin(\Psi_f) &\leq  -2\log(1-2^{1-n}\chi(f)) \\
    &\leq 2s-2\log\prod_{i=1}^{s}\left(1+2^{-h_{q_{i}}}\right),
\end{align}
by (\ref{eq:ddmin_chi}).

We now apply the above general technique to hypergraph states on 2D triangulated lattices (depicted in Fig.~\ref{fig:lattice2}) given by (\ref{eq:cczstate}).
\begin{figure}[t]
 \includegraphics[width=\columnwidth]{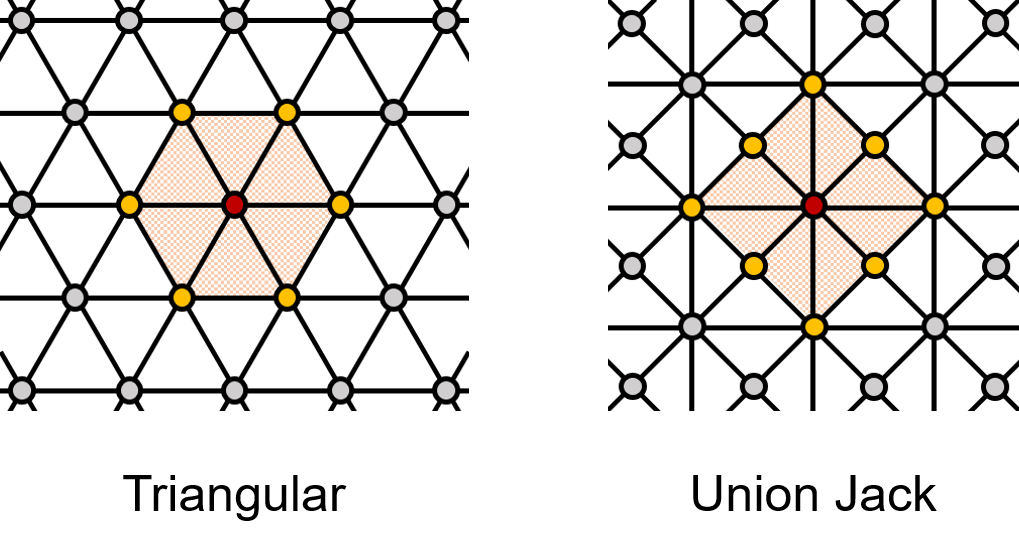}
 \caption{2D triangulated lattices.  For each lattice, the shaded area is the unit cell we choose; the red vertex is the center vertex that only belongs to its cell; the yellow vertices are shared with neighboring cells. The Boolean functions are decomposed  according  to the cells.  \label{fig:lattice2}}
\end{figure}
As we shall see, the decomposition of the cubic characteristic functions is based on the lattice structure.  

First consider the triangular lattice.   As illustrated in Fig.~\ref{fig:lattice2}, consider the lattice to be a tiling of unit cells, each of which is a hexagon (such as the shaded one). 
Label the center vertex of each cell as $x_i$, which is involved in 6 CCZ gates (cubic terms) with the 6 boundary vertices ($x_{i1},\cdots, x_{i6}$) of the cell. So the characteristic function can be expressed as
\begin{align}
    f_{\mathrm{Tri}}(x) =& \sum_i x_i(x_{i1}x_{i2}+x_{i2}x_{i3}+x_{i3}x_{i4}\nonumber\\&+x_{i4}x_{i5}+x_{i5}x_{i6}+x_{i6}x_{i1}) + q,   \label{eq:triangle_decomp}
\end{align}
where $q$ is quadratic, with further constraints that each boundary vertex is shared by 3 neighboring hexagons, so the corresponding variables are actually the same.  Also, boundary vertices never serve as a center vertex.  For simplicity, we make the minor assumption that the lattice is composed of complete cells.   If we consider periodic boundary conditions then, for $m$ cells, there are $6m/3 = 2m$ boundary vertices in total.  Therefore, we have $n = m+2m = 3m$, so (\ref{eq:triangle_decomp}) gives an order-$(s = m = n/3)$ decomposition.  Also note that $q_i = x_{i1}x_{i2}+x_{i2}x_{i3}+x_{i3}x_{i4}+x_{i4}x_{i5}+x_{i5}x_{i6}+x_{i6}x_{i1}$, so, correspondingly
\begin{equation}
    Q_i = \begin{pmatrix}
0 & 1 & 0 & 0 & 0 & 0\\
0 & 0 & 1 & 0 & 0 & 0\\
0 & 0 & 0 & 1 & 0 & 0\\
0 & 0 & 0 & 0 & 1 & 0\\
0 & 0 & 0 & 0 & 0 & 1\\
1 & 0 & 0 & 0 & 0 & 0\\
\end{pmatrix},
\end{equation}
and hence 
\begin{equation}
    h_{q_i} = \frac{1}{2}\mathrm{rank}(Q_i + Q_i^T) = 3.
\end{equation}
Subtituting everything into (\ref{eq:lattice_dmin}), we obtain
\begin{align}
    \ddmax(\Psi_{f_{\mathrm{Tri}}})=  \ddmin(\Psi_{f_{\mathrm{Tri}}}) &\leq \left(\frac{2}{3} - \frac{2}{3}\log\frac{9}{8}\right)n\nonumber\\
    & \lesssim 0.56n.   \label{eq:tri}
\end{align}
For open boundary conditions, the difference is that there are $O(\sqrt{n})$ boundary vertices that are shared by less than 3 cells, which leads to $s = n/3 - O(\sqrt{n})$. As a result, the bound is modified by $-O(\sqrt{n})$, and (\ref{eq:tri}) still holds.

For the Union Jack lattice, we follow a similar procedure.  Now we define the unit cell to be a square involving 9 vertices as illustrated in Fig.~\ref{fig:lattice2}, each of which has 1 center vertex involved in 8 CCZ gates (cubic terms) with the 8 boundary vertices shared with neighboring cells. Again, assume that the lattice is consisted of complete cells. Consider periodic boundary conditions.  Note that among the 8 boundary vertices, 4 in the corner are shared by 4 cells, and 4 on the edge are shared by 2 cells. So for $m$ cells, there are $4m/4 + 4m/2 = 3m$ boundary vertices in total. So $n = m + 3m = 4m$ and thus this cell structure gives $s = n/4$.  It can also be verified that $h_{q_i} = 4$. So, for the corresponding states $\ket{\Psi_{f_{\mathrm{UJ}}}}$, we have
\begin{align}
    \ddmax(\Psi_{f_{\mathrm{UJ}}})=  \ddmin(\Psi_{f_{\mathrm{UJ}}}) &\leq \left(\frac{1}{2} - \frac{1}{2}\log\frac{17}{16}\right)n\nonumber\\
    & \lesssim 0.46n,   \label{eq:uj}
\end{align}
by (\ref{eq:lattice_dmin}).
Similarly, the open boundary conditions lead to a $-O(\sqrt{n})$ correction.
Note that for the Union Jack lattice, another natural definition of the unit cell is the small square with 1 center vertex and 4 corner vertices. However, it can be verified that this cell structure gives $s = \frac12$, which leads to a bound worse than (\ref{eq:uj}).

Finally, we note that we expect the constant factors in the above bounds to be loose, although they already imply that the many-body magic of corresponding SPT phases are much weaker than typical states.  Below we outline two promising paths towards improved bounds: (i) In the above method, the inequality (\ref{eq:ineq}) can be loose, because the shared terms in the $q_i$'s corresponding to neighboring cells will be canceled out in the summation, which leads to the general effect that $h_{\sum q_i} < \sum h_{q_i}$.  Therefore, a direct possibility is a more refined calculation of (\ref{eq:ineq1}) by analyzing $h_{\sum_{i=1}^s r_i q_i}$.
(ii) Show that the characteristic cubic functions are \emph{separable} \cite{Cubic_book, Kolokotronis09:quadratic}, meaning that in the decomposition with the smallest possible $\mathcal{X}$ under all affine transformations of the variables, each cubic term involves exactly one variable from $\mathcal{X}$.
Then, by Theorem 4 of Ref.~\cite{Cubic_book}, the maximum nonquadracity of
such separable cubic functions is 
\begin{equation}
    \max_{{f\in RM(3,n), f~\text{separable}}}\chi(f) = 2^{n-1} - \frac{1}{2}6^{\lfloor n/3 \rfloor},
\end{equation}
leading to the bound for corresponding third-order hypergraph states $\ket{\hat\Psi}$:
\begin{equation}
  \ddmax(\hat\Psi)=\ddmin(\hat\Psi)\leq \left(2-\frac{2}{3}\log 6\right)n \lesssim 0.28n, 
  \label{eq:n/3_upperbound}
\end{equation}
which is already attained by 
$CCZ^{\otimes \frac{n}{3}}\ket{{+}\ldots{+}}$ without interesting entanglement.  Note that most cubic functions studied before are indeed separable \cite{Kasami,KASAMI1976380,Cubic_book}, so the above strong upper bound may hold quite generally for third-order hypergraph states. We believe that the above cases indeed have separable characteristic functions since the cubic terms are highly regular.

\bibliography{t}

\end{document}